\documentclass[11pt,a4paper]{elsarticle}
\usepackage[utf8]{inputenc}

\usepackage{a4wide}
\usepackage{microtype}
\usepackage{amsfonts}
\usepackage{amsmath}
\usepackage{amsthm}
\usepackage{stmaryrd}
\SetSymbolFont{stmry}{bold}{U}{stmry}{m}{n}
\usepackage{dsfont}
\usepackage{xcolor}
\usepackage[linesnumbered,ruled,vlined]{algorithm2e}
\usepackage{hyperref}

\SetCommentSty{mycommfont}
\SetKwInput{KwData}{Input}

\newtheorem{theorem}{Theorem}
\newtheorem{lemma}[theorem]{Lemma}
\newtheorem{definition}{Definition}

\newcommand{\I}{{\mathrm{I}}} 

\newcommand{\PFOOne}{\mbox{$\mathrm{P}_1$}}

\newcommand{\FO}{\mbox{\rm FO}}

\newcommand{\LOneHartig}{\mbox{$\mathcal{L}_1[\I]$}}
\newcommand{\FOOne}{\mbox{$\mathrm{FO}^1$}}
\newcommand{\FOt}{\mbox{$\mathrm{FO}^2$}}
\newcommand{\COne}{\mbox{$\mathrm{C}^1$}}
\newcommand{\GFt}{\mbox{\rm GF}^2}
\newcommand{\FOOneMod}{\mbox{$\mbox{\rm FO}_{\rm MOD}^1$}}

\newcommand{\NP}{\textsc{NP}}
\newcommand{\NExpTime}{\textsc{NExpTime}}

\newcommand{\set}[1]{\{#1\}}
\newcommand{\str}[1]{{\mathfrak{#1}}}
\newcommand{\N}{{\mathbb{N}}}
\newcommand{\Z}{{\mathbb{Z}}} 
\newcommand{\Mod}[1]{ {(\text{mod}\ #1)}}
\newcommand{\existsmod}[3]{\exists^{{#1} {#2} {\Mod{{#3}}} }}
\newcommand{\ind}[2]{\mathds{1}_{#1, #2}}
\newcommand{\tp}{\mathit{tp}} 
\newcommand{\deff}{\stackrel{\text{def}}{=}}
\newcommand{\semantics}[1]{ \llbracket #1 \rrbracket }

\begin{document}

\begin{frontmatter}
  \title{One-Variable Logic Meets Presburger Arithmetic}

  \author{Bartosz~Bednarczyk}
  \ead{bartosz.bednarczyk@cs.uni.wroc.pl}
  \address{Institute of Computer Science, University of Wroc\l aw, 
    Joliot-Curie 15, 50-383 Wroc\l aw, Poland}

  \begin{abstract}
    \noindent We consider the one-variable fragment of first-order logic 
    extended with Presburger constraints. The logic is designed in such a 
    way that it subsumes the previously-known fragments extended with counting, 
    modulo counting or cardinality comparison and combines their expressive 
    powers. We prove $\NP$-completeness of the logic by presenting an optimal 
    algorithm for solving its finite satisfiability problem.
  \end{abstract}

  \begin{keyword}
    finite satisfiability \sep computational complexity \sep decidability 
    \sep classical decision problem \sep arithmetics
  \end{keyword}
\end{frontmatter}


\section{Introduction}

It is well-known that first-order logic~$\FO$ cannot describe 
natural quantitative properties like parity or equicardinality of sets. 
To solve this problem one can think about enlarging the language with
special constructs, e.g., generalized quantifiers like counting quantifiers,
modulo counting quantifiers, majority quantifiers or the H{\"a}rtig quantifier.
However additional expressive power often comes with an increase in 
computational complexity. For example consider the two-variable fragment
of first order logic,~$\FOt$. It is known that~$\FOt$ becomes
undecidable when cardinality comparison via H{\"a}rtig or Rescher quantifiers 
is allowed~\cite{GradelOR99}. On the other hand its extension with counting 
quantifiers is decidable~\cite{GradelOR97,PacholskiST00}. The decidability 
status of~$\FOt$ with modulo counting quantifiers is currently unknown. Thus 
there is no hope to obtain a decidable extension of~$\FOt$ which allows 
all of these features. 

In this paper we take a closer look at the one-variable fragment of first-order 
logic, denoted here by~$\FOOne$. The logic is well-understood
and its finite satisfiability is known to be only~$\NP$-complete. 
We are aware of three extensions of~$\FOOne$ that differ in expressive 
power:~$\COne$,~$\FOOneMod$ and~$\LOneHartig$, 
see e.g.~\cite{PrattHartmann08, Bednarczyk-ESSLLI, GradelOR99}. The mentioned 
logics extend~$\FOOne$ with counting quantifiers~$\exists^{\geq k}$,
modulo-counting quantifiers~$\existsmod{=}{a}{b}$ and the so-called
H{\"a}rtig quantifier~$\I$, respectively. The semantics of the first two 
logics is very intuitive. For the third logic we define~$\I(\varphi, \psi)$ 
to be true if the total number of elements satisfying the formula~$\varphi$ 
is the same as the total number of elements satisfying~$\psi$. 
It follows from~\cite{PrattHartmann08} and~\cite{Bednarczyk-ESSLLI} that 
the finite satisfiability problem for~$\COne$ and~$\FOOneMod$ is~$\NP$-complete, 
even when the numbers in quantifiers are written in binary. Moreover, a practical 
algorithm for deciding satisfiability of a fragment of~$\COne$ was implemented 
and tested in~\cite{FingerB17}. For the third logic, namely~$\LOneHartig$ 
from~\cite{GradelOR99}, the authors of the paper stated that the logic is 
decidable but no proof or complexity bounds were given. 

\subsection{Our contribution}

In this article we present a novel logic~$\PFOOne$ which subsumes 
previously known logics with counting or cardinality comparison, 
i.e.,~$\COne$,~$\FOOneMod$ and~$\LOneHartig$ from~\cite{PrattHartmann08, Bednarczyk-ESSLLI, GradelOR99}. 
Moreover the logic allows one to express percentage constraints. As an example 
we can consider a property that the majority of elements of a model satisfies a given formula~$\varphi$.

We obtain a tight~$\NP$ upper bound for~$\PFOOne$. The proof goes via a 
translation of formulae into a system of inequalities and closely follows the 
techniques presented in~\cite{PrattHartmann08}. However some technical details
differ. As a by-product we fill a gap concerning the complexity of~$\LOneHartig$.


\subsection{Our motivations}

Our main motivation is to see what the scope of the technique of 
Pratt-Hartmann~\cite{PrattHartmann08} is for deciding finite satisfiability 
for~$\COne$. Moreover, we would like to see how powerful a logic we can obtain 
while keeping the complexity reasonably low. Last but not least, the proposed 
logic~$\PFOOne$ is the core part of Presburger Modal Logic~\cite{DemriL10} and 
its~$\NP$-completeness can be used to establish complexities for reasoning tasks
of the family of Euclidean Presburger Modal Logics. A slight generalisation of
the translation from~\cite{KazakovP09} shows that their local and global 
satisfiability can be easily reduced to the finite satisfiability of~$\PFOOne$.

We recently learned about the existence of QFBAPA~\cite{KuncakR07}, 
the quantifier-free fragment of Boolean Algebra with Presburger Arithmetic.
The logics can express similar properties and share the same complexity of 
satisfiability, namely~$\NP$-completeness. Nevertheless, we strongly 
believe that the logic~$\PFOOne$ is a very natural logic, arguably more 
elegant than QFBAPA, and the proof technique used here is much easier to understand.


\section{Preliminaries}

We employ the standard terminology from model theory and linear algebra.
We refer to structures with fraktur letters, and to their universes with 
the corresponding Roman letters. We always assume that structures have
non-empty universes. Here we are interested in \emph{finite} structures 
over a countable \emph{signature}~$\Sigma$ consisting of 
\emph{unary} relational symbols only.

Let~$\mathcal{L}$ be an arbitrary logic. In the 
\emph{finite satisfiability problem} for the logic~$\mathcal{L}$ we ask
whether an input formula~$\varphi$ from~$\mathcal{L}$ is finitely satisfiable, 
i.e., has a finite \emph{model}.


\subsection{Linear algebra and integer linear programming}

By~$\Z_{n}$ we denote the set of all remainders
modulo~$n$, that is the set~$\set{0,1, \ldots, n{-}1}$.

A \emph{linear inequality} is an expression of the form~$t \geq t'$,
where~$t$ and~$t'$ are linear terms. In this paper we are interested
only in linear inequalities with integer coefficients (written in binary). 
It is well known that solving systems of such inequalities over~$\N$ is 
in~$\NP$~\cite{BoroshandTreybig}. 

The following ``sparse solution'' lemma provides an upper bound on the minimum 
number of non-zero unknowns in solutions of systems of linear inequalities:

\begin{lemma}[\cite{AlievLEOW18}] \label{lemma:smallsolution} 
Let~$\mathcal{E}$ be a system of~$\mathit{I}$ inequalities with integer 
coefficients such that the absolute value of each coefficient from~$\mathcal{E}$ 
is bounded by~$\mathit{C}$. If~$\mathcal{E}$ has a solution over~$\N$, then it 
has also a solution over~$\N$ with the number of non-zero unknowns bounded 
by~$2 \mathit{I} \log{\left( 2 \mathit{C} \sqrt{\mathit{I}} \right)}$.
\end{lemma}


\subsection{Syntax of the logic~\texorpdfstring{$\PFOOne$}{POne}}

In this article we propose an extension of a function-free one-variable 
fragment of first-order logic with \textit{counting terms} and 
\textit{Presburger constraints}. We let~$\PFOOne$ denote the formalism. 

The main ingredients of formulae of~$\PFOOne$ are counting 
terms~$t_x$ \cite{GradelOR99}. 
Their intuitive role is to count the total number of witnesses of a given 
formula featuring a single variable~$x$. Such terms can be multiplied by
integer constants and added to each other.
On the top level we allow for the comparison of values of counting-terms with 
a given threshold using a greater-than operator~$\geq$ and to test 
congruence modulo some number~$k$ using~$\equiv_k$. 
More general formulae can be constructed with Boolean connectives 
and by means of nesting.

The minimal syntax of the logic~$\PFOOne$ is given by the following BNF grammar:
$$t_x ::= t_x + t_x \ | \ a \cdot \sharp_x[\varphi(x)]$$
$$
  \varphi ::= 
    P(x) \ | \ \neg \varphi \ | \ \varphi \wedge \varphi \ |
    \ t_x \geq b \ | \ t_x \equiv_c d
$$
where~$P \in \Sigma$ is a unary relational symbol,~$a \in \Z \setminus \{ 0 \}$ 
is a non-zero integer,~$b \in \N$ is a natural number,~$c \in \Z_+$ is a 
positive integer and~$d \in \Z_c$ is a remainder modulo~$c$.

A counting term of the form 
$a_1 \cdot \sharp_x[\varphi_1] + \ldots +
a_n \cdot \sharp_x[\varphi_n]$ 
is abbreviated by~$\Sigma_{i=1}^{n} a_i \cdot \sharp_x[\varphi_i]$.
Note that all standard logical connectives such as~$\vee$, 
$\rightarrow$,~$\leftrightarrow$ as well as other (in)equality 
symbols like~$<, >, \leq$ and~$=$ can be easily defined using Boolean 
combinations and constants. Hence, we will use them as abbreviations. 

We write~$|\varphi|$ to denote the length of a formula~$\varphi$, i.e., 
the number of bits required to encode~$\varphi$ as a string. We will 
assume that all numbers appearing in~$\varphi$ are written in binary. 

\subsection{Semantics of the logic~\texorpdfstring{$\PFOOne$}{POne}}

The semantics of the logic~$\PFOOne$ is a straightforward extension of 
the semantics of first-order logic. For formulae~$\varphi$ not involving
counting terms, the semantics~$\semantics{\varphi}^{\str{M}}$ of~$\varphi$ in a 
model~$\str{M}$ is the same as in first-order logic. We extend it to 
counting terms by defining~$\semantics{\sharp_x[\varphi(x)]}^{\str{M}}$ to be
the cardinality of the set~$\{ a \in M \mid \str{M} \models \varphi[a] \}$.
Addition, multiplication by a constant and comparison are treated in the
obvious way.

\subsection{Expressive power}

We note here that~$\PFOOne$ trivially extends the one-variable fragment of 
first-order logic. Moreover, the logic can capture a scenario of threshold 
counting~$\exists^{\geq k}\varphi(x)$ (i.e.,~$\COne$ from~\cite{PrattHartmann08}) 
as well as modulo counting~$\existsmod{=}{a}{b}\varphi(x)$ (i.e.,~$\FOOneMod$ 
from~\cite{Bednarczyk-ESSLLI}). The logic also allows cardinality comparison, 
i.e., it can simulate the so-called H{\"a}rtig and Rescher quantifiers 
from~\cite{GradelOR99} and percentage constraints, 
e.g.~$\I{x}.(\varphi(x), \psi(x))$ can be encoded 
as~$\sharp_x[\varphi(x)] - \sharp_x[\psi(x)] = 0$. 
Hence~$\PFOOne$ can even express some second-order properties.

\subsection{Types and normal forms}

Let~$\tau$ be a finite signature, and following a standard terminology, we 
define an atomic~$1$-type over~$\tau$ as a maximal satisfiable set of atoms 
or negated atoms  involving only the variable~$x$. Usually we identify a 
$1$-type  with the conjunction of all its elements. We note here that the 
number of all atomic~$1$-types is exponential in the size of~$\tau$. 

When a formula~$\varphi$ is fixed, we often refer to its signature 
(i.e., the set of unary symbols occurring in~$\varphi$) with~$\tau_{\varphi}$. 
Then, the set of all~$1$-types over~$\tau_{\varphi}$ is 
denoted by~$\tp_{\varphi}$ and we refer to its elements 
with~$\pi_1^{\tau_\varphi}, \pi_2^{\tau_\varphi}, 
\ldots, \pi_{|\tp_{\varphi}|}^{\tau_\varphi}$. Additionally, when both a model
$\str{M}$ and a~$1$-type~$\pi$ are fixed, we define~$|\pi|_{\str{M}}$ as the
total number of elements from a structure~$M$ satisfying a~$1$-type~$\pi$.

\begin{definition} \label{def:pfone_normal_form}
We say that a formula~$\varphi \in \PFOOne$ is \emph{flat}, if:
$$\varphi = \bigwedge_{i=1}^{n} 
\left( \sum_{j=0}^{n_i} a_{i,j} \cdot \sharp_x[\varphi_{i,j}] \right)
\bowtie_i b_i
$$
where~$\bowtie_i$ is a comparison symbol, i.e., 
$\bowtie_i\in \set{\leq, \geq, \equiv_k \mid  k \in \N }$,
each~$a_{i,j} \in \Z \setminus \{ 0 \}$ is a non-zero integer, each~$b_i \in \N$
is a natural number and all formulae~$\varphi_{i,j}$ are free of counting terms.
\end{definition}

The main purpose of introducing a flat form for~$\PFOOne$ formulae is to
avoid nesting of counting terms and to simplify reasoning about satisfaction of 
a formula. The following lemma shows that every satisfiable 
$\PFOOne$ formula can be flattened in~$\NP$:

\begin{lemma} \label{lemma:flattening}
There exists a non-deterministic polynomial time procedure, taking as its input
a~$\PFOOne$ formula over a signature~$\tau$ and producing a flat formula
$\varphi'$ over the same signature~$\tau$, such that~$\varphi$ is
satisfiable iff the procedure has a run producing a satisfiable~$\varphi'$.
\end{lemma}

\begin{proof}[Sketch of proof.]
The proof goes in a standard fashion, similarly to the proof of Theorem~$1$ 
in~\cite{PrattHartmann08}. The main idea of the algorithm is to take the 
innermost expression~$e$, from the original formula~$\varphi$, of the
form~$\Sigma a_i \sharp_x[\varphi_i] \geq a$ or~$\Sigma a_i \sharp_x[\varphi_i] \equiv_b c$. 
Since we are designing an~$\NP$ procedure and an expression~$e$ speaks only 
globally about the total number of elements, we can guess whether~$e$ is 
satisfied or not. Then depending on a guess we replace~$e$ with~$\top$ or 
$\bot$ and we put, respectively,~$e$ or~$\neg e$ in front of the formula. 
Additionally, in the case when~$\neg e$ contains a modulo constraint, we 
guess a proper remainder~$c'$ and replace~$\neg \Sigma a_i \sharp_x[\varphi_i] \equiv_b c$ 
with~$\Sigma a_i \sharp_x[\varphi_i] \equiv_b c'$. 
We repeat the whole process until we obtain a flat formula.
\end{proof}


\section{The finite satisfiability of~\texorpdfstring{$\PFOOne$}{POne}}

In this section we will show that the one-variable fragment of first-order
logic remains~$\NP$-complete even if we extend it with Presburger constraints. 
As we mentioned in the beginning of the paper, we are interested 
only in finite models since e.g. modulo constraints do not make 
sense over infinite structures. Our proof will strongly rely on 
techniques presented in~\cite{PrattHartmann08}, namely 
reducing our problem to integer linear programming. 

\subsection{Overview of the method}

Throughout this section, we fix a satisfiable~$\PFOOne$ formula~$\varphi$. 
Due to Lemma \ref{lemma:flattening} we can always produce a flat version 
of~$\varphi$, thus we assume that~$\varphi$ is flat. 

We will first sketch our approach. A crucial observation leading to a simple 
description of~$\PFOOne$ models is that the logic cannot speak about any kind 
of connection between two distinct elements of a model. 
Thus any model~$\str{M}$ of~$\varphi$ can be described up to 
isomorphism by the information about the total number of elements of 
given~$1$-types. We call such information a \emph{characteristic} 
\emph{vector}~$\chi_{\varphi}$. It could be defined in the following way:
$$
\chi_{\varphi} \; \deff \; \left( 
|\pi_0^{\tau_\varphi}|_{\str{M}}, \; |\pi_1^{\tau_\varphi}|_{\str{M}}, \; 
\ldots, \; |\pi_{|\tp_{\varphi}|}^{\tau_\varphi}|_{\str{M}}
\right),
$$
where the~$i$-th element of~$\chi_{\varphi}$ is simply the total 
number of elements from~$M$ of the $i$-th $1$-type.

Our goal is to translate a formula~$\varphi$ into a system of inequalities 
and congruences~$\mathcal{E}$, whose solution will be a tuple~$\chi_{\varphi}$.
Then, we will get rid of congruences, i.e., replace each of them with 
inequalities, at the expense of introducing polynomially many fresh variables.
The obtained system~$\mathcal{E}'$, as well as some of its coefficients, 
will be exponential due to the binary encoding of numbers. Since integer 
linear programming is in~$\NP$~\cite{BoroshandTreybig} we will obtain 
an~$\NExpTime$ upper bound. To improve the complexity of the algorithm, we 
will use Lemma~\ref{lemma:smallsolution}, which states that if there is a 
solution for~$\mathcal{E}$, there is also a ``sparse'' solution, i.e., 
assigning only polynomially many non-zero values to unknowns. 

It is worth pointing out that due to the presence of exponential coefficients
we cannot easily adapt the lemma about small solutions from~\cite{PrattHartmann08}. 
The technique we use, namely Lemma~\ref{lemma:smallsolution}, is more 
sophisticated and requires a more difficult proof. We will use it as a black box.

\subsection{A translation into a system of inequalities and congruences} \label{sec:trans}

We are going to describe a potential model~$\str{M}$ of the formula~$\varphi$
in terms of unknowns and inequalities. In the desired system of inequalities,
we will have exponentially many variables~$x_k$, where each~$x_k$ corresponds
to~$|\pi_k|_{\str{M}}$ in a characteristic vector and each inequality or 
congruence corresponds to a threshold given in some conjunct from~$\varphi$.

Let~$\varphi_i$ be the~$i$-th conjunct from~$\varphi$, i.e., 
$\varphi_i = \left( \sum_{j=0}^{n_i} a_{i,j} 
\cdot \sharp_x[\varphi_{i,j}] \right) \bowtie_i b_i$.
Then, for every~$1$-type~$\pi_k$ we will associate an indicator~$\ind{i,j}{k}$,
whose intuitive role will be to tell us whether the~$k$-th type~$\pi_k$ is 
compatible with the formula~$\varphi_{i,j}$. More formally:
$$
\ind{i,j}{k} {=} \left\{\begin{matrix} 1, & \text{if} \;
\models \pi_{k} \rightarrow \varphi_{i,j} \\ 0, & \text{otherwise}
\end{matrix}\right.
$$

With the above definition it is not hard to see that the value of 
a counting term~$\sharp_x[\varphi_{i,j}]$ is equal to 
$\Sigma_{k=1}^{|\tp_{\varphi}|} \ind{i,j}{k} \cdot x_k$. By multiplying
such value with constants~$a_{i,j}$ and summing it over~$j$, the whole
formula~$\varphi_i$ can be represented as the following inequality or congruence:
$$
\left( \sum_{j=0}^{n_i} a_{i,j} \cdot 
\left( \Sigma_{k=1}^{|\tp_{\varphi}|} \ind{i,j}{k} \cdot x_k  \right)
\right) \bowtie_i b_i
$$
After rearranging the left-hand side of the above expression, 
we obtain a linear term with unknowns~$x_1, x_2, \ldots, x_{|\tp_{\varphi}|}$. 
Note that coefficients in front of variables~$x_k$ are exponential due to the
binary encoding. We construct a system of inequalities and 
congruences~$\mathcal{E}_{\varphi}$ by translating each conjunct~$\varphi_i$ 
from~$\varphi$ in the presented way.

The following lemma follows directly from the fact that 
each model~$\mathfrak{M}$ of~$\PFOOne$ formula can be described up 
to isomorphism by a characteristic vector and from the construction 
of~$\mathcal{E}_{\varphi}$. 

\begin{lemma}\label{lemma:vec1}
Each solution of~$\mathcal{E}_{\varphi}$ is a characteristic vector 
of some model~$\mathfrak{M}$ of a~$\PFOOne$ formula~$\varphi$.
\end{lemma}

\subsection{Getting rid of congruences}

The obtained system~$\mathcal{E}_{\varphi}$ can still contain linear terms with 
congruences. We will show a way how to replace them with inequalities. 
Let us assume that the~$i$-th equation of the system~$\mathcal{E}_{\varphi}$ is 
a congruence of the following form:

$$
a_1^{i} \cdot x_1 + a_2^{i} \cdot x_2 + \ldots + 
a_{|\tp_{\varphi}|}^{i} \cdot x_{|\tp_{\varphi}|} \equiv_{k_i} b_i
$$

For any natural number~$S_i$, there exists a \emph{remainder}~$r_i \in \Z_{k_i}$
and a \emph{quotient}~$q_i \in \N$, such that~$S_i = r_i + q_i k_i$. 
Thus we only need to ensure that the remainder~$r_i$ is equal to~$b_i$.
Since we do not know the precise value of the quotient~$q_i$, we introduce
a fresh variable~$y_i$ to represent it. We can rewrite the above 
congruence as~$\sum_{j=1}^{|\tp_{\varphi}|} a_j^i = b_i + k_i \cdot y_i$, 
which is equivalent to:
$$
a_1^{i} \cdot x_1 + a_2^{i} \cdot x_2 + \ldots + 
a_{|\tp_{\varphi}|}^{i} \cdot x_{|\tp_{\varphi}|} - b_i - k_i \cdot y_i \leq 0,
$$
$$
a_1^{i} \cdot x_1 + a_2^{i} \cdot x_2 + \ldots + 
a_{|\tp_{\varphi}|}^{i} \cdot x_{|\tp_{\varphi}|} - b_i - k_i \cdot y_i \geq 0,
$$

Let~$\mathcal{E}_{\varphi}'$ be the system of inequalities obtained 
from~$\mathcal{E}_{\varphi}$ by exhaustive elimination of all congruences. 
Since each step of the ``congruence-elimination'' procedure described 
above is sound, together with Lemma~\ref{lemma:vec1} we establish:

\begin{lemma}\label{lemma:vec2}
Each solution of~$\mathcal{E}_{\varphi}'$ is a characteristic vector 
of some model~$\mathfrak{M}$ of a~$\PFOOne$ formula~$\varphi$.
\end{lemma}

One can observe that the number of
equations in~$\mathcal{E}_{\varphi}'$ is bounded by~$2n$ (i.e., where~$n$ 
is the number of conjuncts from flat~$\varphi$), which is clearly of 
polynomial size in~$|\varphi|$. Integer coefficients of the 
system~$\mathcal{E}_{\varphi}'$ can be bounded by the sum of the absolute 
values of the numbers occurring in the formula~$\varphi$. Since every number 
can be exponential in~$|\varphi|$ (due to the binary encoding) and the 
mentioned sum contains at most polynomially many elements, we can conclude 
that each coefficient from the system~$\mathcal{E}_{\varphi}'$ is 
bounded exponentially in~$|\varphi|$.

\subsection{Algorithm}

By using Lemma~\ref{lemma:smallsolution} we know that the minimum number of 
non-zero unknowns in a sparse solution of~$\mathcal{E}_{\varphi}'$ can be 
bounded by a polynomial function of~$|\varphi|$. Hence we 
non-deterministically guess which unknowns will be non-zero and we construct 
a corresponding system~$\mathcal{E}_{\varphi}''$ directly for them. The obtained 
system has polynomial size in~$|\varphi|$, thus it is solvable in~$\NP$.

Below we present a non-deterministic polynomial time algorithm for testing
whether a given~$\PFOOne$ formula has a finite model.

\begin{algorithm}[h] 
  \label{algo:sat_test_pfoone}
  \DontPrintSemicolon
  \KwData{A formula~$\varphi \in \PFOOne$} 
  \caption{Satisfiability test for~$\PFOOne$}

  \textbf{guess}~$\varphi'$ -- a flat version of~$\varphi$ \tcp*{in~$\NP$, Lemma \ref{lemma:flattening}}

  \textbf{guess} which~$1$-types are realized at least once. \tcp*{polynomially many, Lemma \ref{lemma:smallsolution} } 

  Write the system of inequalities~$\mathcal{E}_{\varphi}''$ for the guessed~$1$-types. \tcp*{of poly size} 

  Return \textit{True} iff ~$\mathcal{E}_{\varphi}''$ has a solution over~$\N$. \tcp*{in~$\NP$ \cite{BoroshandTreybig}}
\end{algorithm}

To ensure the correctness of the algorithm, we prove the following lemma:
\begin{lemma}
  A~formula $\varphi \in \PFOOne$ has a finite model if and only if
  Procedure~\ref{algo:sat_test_pfoone} returns \textit{True}.
\end{lemma}

\begin{proof}
We first assume that an input formula~$\varphi$ has a finite model. Therefore, we can
obtain a flat finitely satisfiable formula~$\varphi'$ (by Lemma~\ref{lemma:flattening}) 
and describe its model in terms of linear inequalities and congruences 
$\mathcal{E}_{\varphi'}$ (by Lemma~\ref{lemma:vec1}). 
Clearly the system has a solution over~$\N$ (e.g., a characteristic 
vector~$\chi_{\varphi'}$), hence also suitable choices 
for~$\mathcal{E}_{\varphi'}'$ and~$\mathcal{E}_{\varphi'}''$ have solutions. 
Hence Procedure \ref{algo:sat_test_pfoone} returns \textit{True}.

Conversely, suppose that Procedure \ref{algo:sat_test_pfoone} 
returns \textit{True} for its input formula~$\varphi$.  
We construct a model for~$\varphi$. We do it simply by 
taking a proper number of realizations of each~$1$-type, exactly as described 
in the solution of the constructed system of linear inequalities~$\mathcal{E}_{\varphi'}''$.
\end{proof}

Using the above lemma, one can conclude the following theorem:

\begin{theorem}
The finite satisfiability problem for~$\PFOOne$ is~$\NP$-complete.   
\end{theorem}

\begin{proof}
The lower bound comes trivially from Boolean satisfiability problem or 
from the earlier works on~$\COne$ \cite{PrattHartmann08}. For the upper bound 
it is enough to note that Procedure~\ref{algo:sat_test_pfoone} works in~$\NP$.
It follows from (i) the fact that flattening can be done in~$\NP$ 
(Lemma~\ref{lemma:flattening}), (ii) correspondence between systems of
inequalities and characteristic vectors of~$\PFOOne$~models (Lemma~\ref{lemma:vec1}), 
(iii) existence of sparse solutions of systems of inequalities 
(Lemma \ref{lemma:smallsolution}), and~(iv) an~$\NP$ algorithm for solving 
systems of inequalities with polynomially many unknowns~\cite{BoroshandTreybig}.
\end{proof}


\section{Conclusions and future work}

\subsection{Conclusions}

In this article we proposed a new logic called~$\PFOOne$ which significantly 
increase the expressive power of the one-variable fragment of first-order logic.
The obtained logic generalizes previously known concepts of counting, i.e., 
threshold counting, modulo counting and cardinality comparison. By using a 
generic method of transforming a formula into a system of inequalities, we 
prove that every satisfiable~$\PFOOne$ formula can be represented as a 
system of inequalities of polynomial size. By using a well-known theorem that 
integer linear programming is in~$\NP$ we obtained a tight~$\NP$ upper bound 
for finite satisfiability for the logic~$\PFOOne$. This proves that the 
complexity of~$\PFOOne$ with expressive numerical constraints does not 
differ from the classical one-variable fragment of~$\FO$, or even from
Boolean satisfiability, which is rather surprising.

\subsection{Future work}

For future work we would like to investigate other classical decidable 
fragments of first-order logic and see how their complexity and decidability 
status behaves after adding some form of Presburger constraints. 

One candidate could be the two-variable fragment of first-order logic~$\FOt$. 
However in the presence of cardinality comparison the logic becomes undecidable
\cite{GradelOR99}. 

Another prominent logic is a two-variable fragment of 
the guarded fragment of first-order logic~$\GFt$, which is known to be decidable 
even in the presence of counting quantifiers~\cite{Pratt-Hartmann07}. However, 
even adding modulo constraints to the logic is a challenging task and currently 
we do not even have a decidability proof. 
On the other hand, some decidable fragments of~$\GFt$ extended with
Presburger constraints are known. We already know that the complexity of 
the modal logic~$\mathcal{K}$ or the description logic~$\mathcal{ALC}$ 
do not differ from their Presburger versions, 
see~\cite{Baader17,DemriL10,KupkePS15}.
We believe that to obtain tight complexity bounds for Presburger~$\GFt$ one
should start with a more modest goal, i.e., to establish the exact complexity
of Presburger~$\mathcal{ALCI}$, namely an extension of~$\mathcal{ALC}$ with
inverse relations.


\section*{Acknowledgments}

This work is supported by the Polish Ministry of Science and 
Higher Education program "Diamentowy Grant" no. DI2017 006447. 
The author would also like to thank the two anonymous reviewers 
as well as Witold Charatonik, Emanuel Kiero\'nski and Antti Kuusisto
for their careful proofreading and for pointing out numerous 
grammatical mistakes. 


\bibliographystyle{plain}
\bibliography{bibliography}

\begin{thebibliography}{10}

\bibitem{AlievLEOW18}
Iskander Aliev, Jes{\'{u}}s A.~De Loera, Friedrich Eisenbrand, Timm Oertel, and
  Robert Weismantel.
\newblock {The} {Support} of {Integer} {Optimal} {Solutions}.
\newblock {\em {SIAM} Journal on Optimization}, 28(3):2152--2157, 2018.

\bibitem{Baader17}
Franz Baader.
\newblock {A} {N}ew {D}escription {L}ogic with {S}et {C}onstraints and
  {C}ardinality {C}onstraints on {R}ole {S}uccessors.
\newblock In Clare Dixon and Marcelo Finger, editors, {\em Frontiers of
  Combining Systems - 11th International Symposium, FroCoS 2017,
  Bras{\'{\i}}lia, Brazil, September 27-29, 2017, Proceedings}, volume 10483 of
  {\em Lecture Notes in Computer Science}, pages 43--59. Springer, 2017.

\bibitem{Bednarczyk-ESSLLI}
Bartosz Bednarczyk.
\newblock {O}n {O}ne {V}ariable {F}ragment of {F}irst {O}rder {L}ogic with
  {M}odulo {C}ounting {Q}uantifier.
\newblock In Karoliina Lohiniva and Johannes Wahle, editors, {\em ESSLLI 2017
  Student Session, 29th European Summer School in Logic, Language {\&}
  Information, July 17-28, 2017, Toulouse, France}, pages 7--13, 2017.

\bibitem{BoroshandTreybig}
I.~Borosh, M.~Flahive, and B.~Treybig.
\newblock Small solutions of linear {D}iophantine equations.
\newblock {\em Discrete Mathematics}, 58(3):215--220, 1986.

\bibitem{DemriL10}
St{\'{e}}phane Demri and Denis Lugiez.
\newblock {Complexity} of modal logics with {Presburger} constraints.
\newblock {\em J. Applied Logic}, 8(3):233--252, 2010.

\bibitem{FingerB17}
Marcelo Finger and Glauber~De Bona.
\newblock {Algorithms} for {Deciding} {Counting} {Quantifiers} over {Unary}
  {Predicates}.
\newblock In Satinder~P. Singh and Shaul Markovitch, editors, {\em Proceedings
  of the Thirty-First {AAAI} Conference on Artificial Intelligence, February
  4-9, 2017, San Francisco, California, {USA.}}, pages 3878--3884. {AAAI}
  Press, 2017.

\bibitem{GradelOR97}
Erich Gr{\"{a}}del, Martin Otto, and Eric Rosen.
\newblock {T}wo-{V}ariable {L}ogic with {C}ounting is {D}ecidable.
\newblock In {\em Proceedings, 12th Annual {IEEE} Symposium on Logic in
  Computer Science, Warsaw, Poland, June 29 - July 2, 1997}, pages 306--317.
  {IEEE} Computer Society, 1997.

\bibitem{GradelOR99}
Erich Gr{\"{a}}del, Martin Otto, and Eric Rosen.
\newblock {Undecidability} results on two-variable logics.
\newblock {\em Arch. Math. Log.}, 38(4-5):313--354, 1999.

\bibitem{KazakovP09}
Yevgeny Kazakov and Ian Pratt{-}Hartmann.
\newblock {A} {N}ote on the {C}omplexity of the {S}atisfiability {P}roblem for
  {G}raded {M}odal {L}ogics.
\newblock In {\em Proceedings of the 24th Annual {IEEE} Symposium on Logic in
  Computer Science, {LICS} 2009, 11-14 August 2009, Los Angeles, CA, {USA}},
  pages 407--416. {IEEE} Computer Society, 2009.

\bibitem{KuncakR07}
Viktor Kuncak and Martin~C. Rinard.
\newblock {T}owards {Efficient} {Satisfiability} {Checking} for {Boolean}
  {Algebra} with {Presburger} {Arithmetic}.
\newblock In Frank Pfenning, editor, {\em Automated Deduction - CADE-21, 21st
  International Conference on Automated Deduction, Bremen, Germany, July 17-20,
  2007, Proceedings}, volume 4603 of {\em Lecture Notes in Computer Science},
  pages 215--230. Springer, 2007.

\bibitem{KupkePS15}
Clemens Kupke, Dirk Pattinson, and Lutz Schr{\"{o}}der.
\newblock {R}easoning with {G}lobal {A}ssumptions in {A}rithmetic {M}odal
  {L}ogics.
\newblock In Adrian Kosowski and Igor Walukiewicz, editors, {\em Fundamentals
  of Computation Theory - 20th International Symposium, {FCT} 2015,
  Gda{\'{n}}sk, Poland, August 17-19, 2015, Proceedings}, volume 9210 of {\em
  Lecture Notes in Computer Science}, pages 367--380. Springer, 2015.

\bibitem{PacholskiST00}
Leszek Pacholski, Wieslaw Szwast, and Lidia Tendera.
\newblock {Complexity} {Results} for {First}-{Order} {Two}-{Variable} {Logic}
  with {Counting}.
\newblock {\em {SIAM} J. Comput.}, 29(4):1083--1117, 2000.

\bibitem{Pratt-Hartmann07}
Ian Pratt{-}Hartmann.
\newblock {Complexity} of the {G}uarded {T}wo-variable {F}ragment with
  {C}ounting {Q}uantifiers.
\newblock {\em J. Log. Comput.}, 17(1):133--155, 2007.

\bibitem{PrattHartmann08}
Ian Pratt{-}Hartmann.
\newblock On the {Computational} {Complexity} of the {N}umerically {D}efinite
  {S}yllogistic and related logics.
\newblock {\em Bulletin of Symbolic Logic}, 14(1):1--28, 2008.

\end{thebibliography}
\end{document}